\documentclass{llncs}

\usepackage{graphicx}
\usepackage{amsfonts}
\usepackage{amsmath}
\usepackage{amssymb}
\usepackage{url}
\usepackage{subfig}
\usepackage{hyperref}

\allowdisplaybreaks

\title{On the intersection graph of the disks with diameters the sides of a convex $n$-gon}

\authorrunning{Herrera and P\'erez-Lantero}

\author{
  		Luis H. Herrera\inst{1}
	\and
  		Pablo P\'erez-Lantero\inst{1}
}

\institute{
		Dept.\ de Matem\'atica y Ciencia de la Computaci\'on, Universidad de Santiago, Chile.\\
        \email{\{luis.herrera.b,pablo.perez.l\}@usach.cl}
}

\begin{document}
\maketitle

\begin{abstract}
Given a convex $n$-gon, we can draw $n$ disks (called side disks) where each disk has a different side of the polygon as diameter and the midpoint of the side as its center. The intersection graph of such disks is the undirected graph with vertices the $n$ disks and two disks are adjacent if and only if they have a point in common. Such a graph was introduced by Huemer and P\'erez-Lantero in 2016, proved to be planar and Hamiltonian. In this paper we study further combinatorial properties of this graph. We prove that the treewidth is at most 3, by showing an $O(n)$-time algorithm that builds a tree decomposition of width at most 3, given the polygon as input. This implies that we can construct the intersection graph of the side disks in $O(n)$ time. We further study the independence number of this graph, which is the maximum number of pairwise disjoint disks. The planarity condition implies that for every convex $n$-gon we can select at least $\lceil n/4 \rceil$ pairwise disjoint disks, and we prove that for every $n\ge 3$ there exist convex $n$-gons in which we cannot select more than this number. Finally, we show that our class of graphs includes all outerplanar Hamiltonian graphs except the cycle of length four, and that it is a proper subclass of the planar Hamiltonian graphs.  
\end{abstract}

\section{Introduction}

Let $P$ be a convex polygon of $n$ sides denoted $s_0,s_1,\ldots,s_{n-1}$ counter-clockwise. For $i\in[0\ldotp\ldotp n-1]$, let $D_i$ denote the disk with diameter the length of $s_i$ and center the midpoint of $s_i$. Since $D_i$ is constructed on the side $s_i$ of $P$, we say that $D_i$ is a {\em side disk} of $P$. Let $G(P)$ denote the intersection graph of $\mathcal{D}=\mathcal{D}(P)=\{D_0,D_1,\ldots,D_{n-1}\}$. Given a finite set $\mathcal{D}$ of disks in the plane, the {\em intersection graph} of $\mathcal{D}$ is the undirected graph $G(\mathcal{D})=(V,E)$, where $V=\mathcal{D}$ and $\{D,D'\}\in E$ if and only if the intersection of $D$ and $D'$ is not empty. See Figure~\ref{fig:example}. In this paper, all graphs are undirected. We say that a graph $G$ is {\em side-disk realizable} if there exists a convex polygon $P$ such that $G\simeq G(P)$. For two graphs $G$ and $H$, we write $G\simeq H$ to denote that $G$ and $H$ are isomorphic.

Given a graph $G=(V,E)$, a {\em tree decomposition} of $G$ is a tree $T$ of nodes $X_1,X_2,\ldots,X_N\subset V$, which satisfies the following three properties:
\begin{itemize}
    \item[(1)] $X_1\cup X_2\cup \dots \cup X_N=V$ (i.e.\ each $v\in V$ is contained in at least one node).
    \item[(2)] For every edge $\{v,w\}\in E$, there exists a node $X_i$ such that $v,w\in X_i$.
    \item[(3)] If nodes $X_i$ and $X_k$ contain a vertex $v\in V$, then all nodes in the 
    shortest path that connects $X_i$ and $X_k$ contain $v$ as well. (i.e.\ the nodes containing $v$ induce 
    a connected subgraph of $T$).
\end{itemize}
The {\em width} of $T$ is $\max_{i\in[1\ldotp\ldotp N]}|X_i|-1$, and the {\em treewidth} of $G$, denoted $tw(G)$, is the minimum width over all tree decompositions of $G$. 

Given an undirected graph $G=(V,E)$, a subset $V'\subseteq V$ is an {\em independent set} of $G$ if no pair of vertices in $V'$ define an edge in $E$. The set $V'$ is a {\em maximum independent set} (MIS) if $V'$ is an independent set of maximum cardinality among all independent sets of $G$. Let $\alpha(G)$ denote the {\em independence number} of $G$, which is the cardinality of a MIS of $G$. Then, any independent set of $G(\mathcal{D})$ is a subset of pairwise disjoint disks of $\mathcal{D}$. 


\begin{figure}[t]
	\centering
	\includegraphics[scale=0.7,page=7]{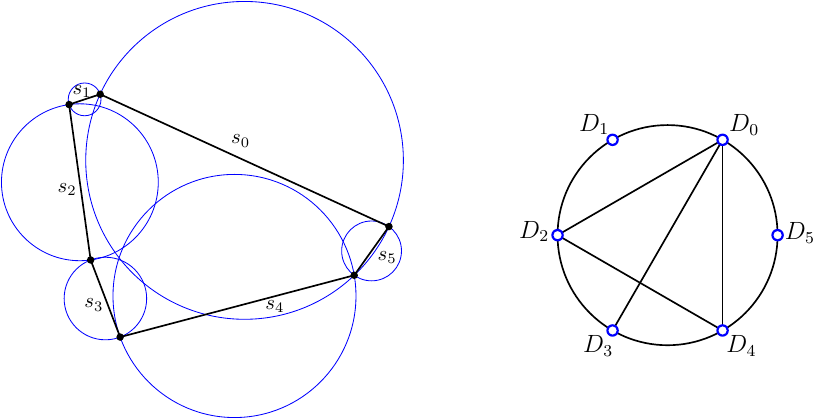}
	\caption{\small{Left: An hexagon of sides $s_0,\ldots,s_5$. Right: The intersection graph
	of the side disks $D_0,\ldots,D_5$.}}
	\label{fig:example}
\end{figure}

We study the graph $G(P)$, which was introduced by Huemer and P\'erez-Lantero~\cite{huemer2016}, motivated by the question whether in any convex pentagon there is a pair of disjoint side disks. The graph $G(P)$ is Hamiltonian, with the Hamiltonian cycle $(D_0,D_1,\ldots,D_{n-1},D_0)$, and they proved that for every convex polygon $P$ of $n\ge 3$ sides, $G(P)$ is also planar. Our results are the following ones:

\begin{enumerate}
    \setlength\itemsep{1em}
    
    \item We present an $O(n)$-time algorithm that receives as input the polygon $P$, and 
    returns a tree decomposition of $G(P)$ of width at most 3. This implies $tw(G(P))\le 3$. Using
    this tree decomposition, we can identify all edges of $G(P)$ in $O(n)$ time as well
    (Section~\ref{sec:treewidth}).
    
    \item The famous 4-color theorem implies $\alpha(G(P))\ge \lceil n/4\rceil$. We prove 
    that this bound is tight: For every $n\ge 3$, there exist convex polygons $P$ of $n$ 
    sides for which $\alpha(G(P))=\lceil n/4\rceil$
    (Section~\ref{sec:size-MIS}).
    
    \item We prove that all outerplanar Hamiltonian graphs, except the cycle of length four,
    are side-disk realizable,
    there exist convex polygons $P$ such that $G(P)$ is not outerplanar, and 
    there exist planar Hamiltonian graphs which are not side-disk realizable.
    Then, the class of the intersection graphs of the side disks is a 
    proper subclass of the planar Hamiltonian graphs
    (Section~\ref{sec:classes}).
\end{enumerate}

It is known that a graph has treewidth at most 2 if and only if it does not contain the complete graph $K_4$ as a minor~\cite{bodlaender1994tourist}. Since there exist convex polygons $P$ such that $G(P)$ contains $K_4$ as a minor (see for example Section~\ref{sec:size-MIS}, or Figure~\ref{fig:example}), the bound $tw(G(P))\le 3$ is tight. That is, such polygons $P$ satisfy $tw(G(P))=3$. More importantly, the graph $G(P)$ is planar and Hamiltonian, and in general planar Hamiltonian graphs can have unbounded treewidth. For example, the $k\times k$ grid graph of $n=k^2$ vertices, where $k$ is even, is both planar and Hamiltonian and its treewidth is $k=\sqrt{n}$.

Since every polygon $P$ satisfies $tw(G(P))\le 3$, we can use a known linear-time algorithm~\cite{bodlaender1996linear} to find a tree decomposition of $G(P)$ of width at most $3$. On one hand, the asymptotic running time of this algorithms encloses a big hidden constant. On the other hand, this algorithm uses explicitly the edges of the graph $G(P)$ (there are at most $3n-6$ edges~\cite{huemer2016}), which need to be obtained from the pairwise intersections between the side disks of $P$ (e.g.\ in $O(n\log n)$ time with a plane sweep~\cite{swart1983}). In our result, a tree decomposition $T$ of $G(P)$ of small width is found directly from $P$ in $O(n)$ time, using the medial axis of $P$ as {\em guide}. Note that $T$ has $O(n)$ nodes, each of at most 4 vertices. The edges of $G(P)$ can be obtained by for each node of $T$ querying for the adjaceny of each pair of vertices of the node (which are side disks). Since $T$ has $O(n)$ nodes, and we perform at most $\binom{4}{2}=6$ adjaceny queries in each of them, the edges of $G(P)$ can be obtained from $P$ in $O(n)$ time. 

The fact that $tw(G(P))\le 3$ (i.e.\ the treewidth is bounded) implies that a large class of NP-hard graph problems (such as maximum independent set, minimum dominating set, hamiltonian cycle, chromatic number, partition into triangles, etc.) can be solved in polynomial time in the graph $G(P)$, using dynamic programming~\cite{bodlaender1988dynamic,bodlaender1994tourist}. For example, if $w\le 3$ is the width of the tree decomposition given by our algorithm, finding a maximum independent set (MIS) of $G(P)$ can be done in $O(2^w\cdot n)=O(n)$ time~\cite{bodlaender1994tourist}. This known algorithm for finding a MIS, and similar algorithms for other NP-hard graph problems, are belived to be time optimal in graphs of bounded treewidth~\cite{Lokshtanov2011}. In general, finding a MIS in a planar Hamiltonian graph is NP-hard~\cite{fleischner2010}, as well as finding a MIS in an intersection graph of disks, even if the disk representation is given as input~\cite{clark1990,wang1988}. 

{\bf Further notation}: Given three different points $p$, $q$, and $r$ in the plane, let $pq$ be the line segment with endpoints $p$ and $q$, $\Delta pqr$ be the triangle with vertex set $\{p,q,r\}$, and $\angle pqr$ be the angle not bigger than $\pi$ with vertex $q$ and sides through $p$ and $r$, respectively. Given a segment $s$, let $|s|$ be the length of $s$, $\ell(s)$ the line that contains $s$, and $D_s$ the disk with diameter $|s|$ and center the midpoint of $s$. Given a pair $D,D'$ of disks, we say that $D$ and $D'$ are {\em intersecting} if $D\cap D'\neq \emptyset$, and that $D,D'$ has a {\em proper} intersection if the interiors of $D$ and $D'$ have a non-empty intersection. We say that a set of disks is {\em proper} if every pair of intersecting disks in the set has a proper intersection.

\section{Bound of treewidth}\label{sec:treewidth}

In this section, we present a linear-time algorithm that constructs from the polygon $P$ a tree decomposition of $G(P)$ of width at most 3. Our algorithm uses (as a preprocessing) the medial axis of $P$. The {\em medial axis} of a simple polygon of $n$ sides is the locus of the points of the polygon that have more than one closest point in the boundary~\cite{preparata1977}, and can be computed in $O(n)$ time~\cite{chin1999finding}. If the polygon is convex, the medial axis is a tree made of line segments, each contained in the bisector of two sides (see Figure~\ref{fig:medial-axis}). 

Given a convex polygon, a {\em maximal} disk is a disk contained in the polygon and tangent to (at least) three sides. Important properties of the medial axis in a convex polygon are as follows: 
\begin{itemize}
    \item Every maximal disk is centered at a vertex of the medial axis.
    
    \item If a maximal disk is tangent to three different sides $a,b,c$ and 
    its center is denoted by $u$, then the center $u'$ of another maximal disk
    tangent to two sides in $\{a,b,c\}$ and another side $d\notin \{a,b,c\}$,
    is such that $u=u'$ or the segment $uu'$ is an edge of the medial axis. 
    
    \item The medial axis can be computed in $O(n)$ time so that each vertex $u$
    is associated with the sides of the polygon (there are at least three) tangent
    to the maximal disk centered at $u$.
\end{itemize}

\begin{figure}[t]
	\centering
	\includegraphics[scale=0.75,page=2]{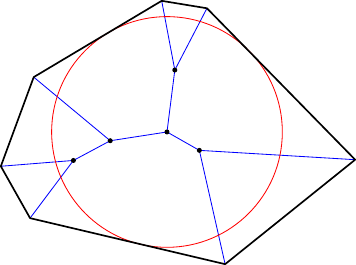}
	\caption{\small{
		Medial axis of a convex polygon and the maximal disk of maximum radius.
	}}	
	\label{fig:medial-axis}
\end{figure}

Let $D_n$ denote $D_0$. Given $u,v\in[0\ldotp\ldotp n]$ with $u\le v+1$, let $D[u,v]$ denote the set $\{D_u,D_{u+1},\ldots,D_{v}\}$ of consecutive side disks. Note that if $u=v+1$ then $D[u,v]$ is the empty set. Two sets $\mathcal{D}_1,\mathcal{D}_2\subset D[0,n-1]=\mathcal{D}$ are called {\em independent} if every disk in $\mathcal{D}_1$ is disjoint from every disk in $\mathcal{D}_2$. 

\begin{lemma}\label{lem:1-chord}
Let $s_i$ and $s_j$, $i<j-1$, be two sides of $P$ such that: there exists a maximal disk $Q\subset P$ tangent to them, and either the lines $\ell(s_i)$ and $\ell(s_j)$ are parallel or for every $k\in[i+1\ldotp\ldotp j-1]$ the line $\ell(s_k)$ separates the intersection point of $\ell(s_i)$ and $\ell(s_j)$ from the interior of $P$. Then, the only disk in the set $\mathcal{D}\setminus D[i,j]$ that can intersect a disk in $D[i+1,j-1]$ is the disk $D_t$, $t\in[0\ldotp\ldotp n]\setminus [i\ldotp\ldotp j]$, such that the interior of $s_t$ is intersected by the 
bisector of $\ell(s_i)$ and $\ell(s_j)$ through $P$.
%
%
%
%
%
%
%
\end{lemma}

Lemma~\ref{lem:1-chord} can be proved with the arguments used by Huemer and P\'erez-Lantero~\cite{huemer2016} to prove the 1-Chord Lemma (see~\cite{huemer2016}).

\begin{lemma}\label{lem:indep}
Let $s_i$, $s_j$, and $s_k$, $i<j<k$, be three sides of $P$ such that: there exists a maximal disk $Q\subset P$ tangent to them, $\ell(s_i)$ and $\ell(s_j)$ intersect at point $A$, $\ell(s_j)$ and $\ell(s_k)$ intersect at point $B$, and $s_j\subset AB$. Then, the sets $D[i+1,j-1]$ and $D[j+1,k-1]$ are independent. Furthermore, the lemma still holds if $\ell(s_i)$ and $\ell(s_j)$ are parallel, or $\ell(s_j)$ and $\ell(s_k)$ are parallel.
\end{lemma}

\begin{proof}
Let $T_i$, $T_j$, and $T_k$ be the points of tangency between $Q$ and $s_i$, $s_j$, and $s_k$, respectively. Let $Q_A$ (resp.\ $Q_B$) denote the disk with center $A$ (resp.\ $B$) and radius $|AT_i|=|AT_j|$ (resp.\ $|BT_j|=|BT_k|$). Note that $Q_A$ and $Q_B$ are tangent and interior disjoint (see Figure~\ref{fig:2circles}). It can be proved that each disk in $D[i+1,j-1]$ is contained in the interior of $Q_A$~\cite[Lemma~7]{huemer2016}. Similary, each disk in $D[j+1,k-1]$ is contained in the interior of $Q_B$. Then, the lemma follows. If $\ell(s_i)$ and $\ell(s_j)$ are parallel, then we can set $Q_A$ as the halfplane bounded by $\ell(T_iT_j)$ that does not contain $s_k$, which will contain every disk in $D[i+1,j-1]$. Similarly, if $\ell(s_j)$ and $\ell(s_k)$ are parallel, we can set $Q_B$ as the halfplane bounded by $\ell(T_jT_k)$ that does not contain $s_i$, which will contain every disk in $D[j+1,k-1]$.\qed
\end{proof}

\begin{figure}[t]
	\centering
	\includegraphics[scale=0.7,page=8]{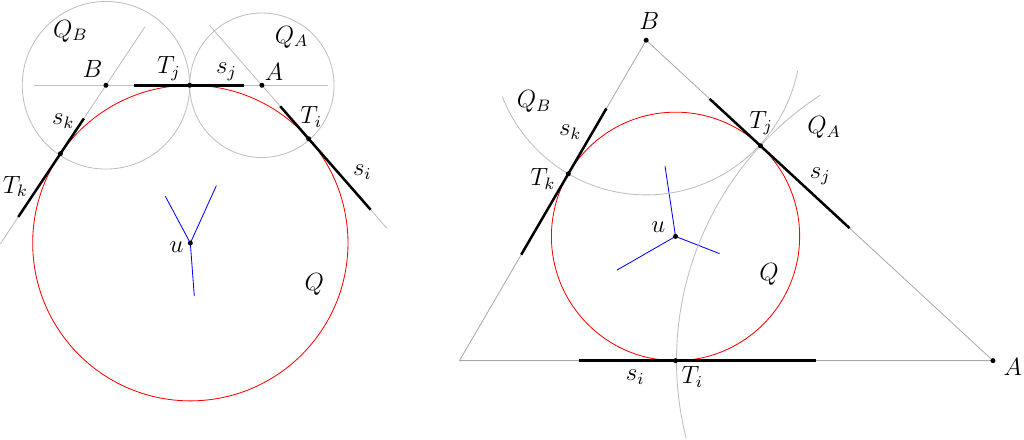}
	\caption{\small{
		Illustration of Lemma~\ref{lem:indep}. In both cases, the center $u$ of $Q$ is a vertex
		of the medial axis of $P$.
	}}	
	\label{fig:2circles}
\end{figure}

\begin{theorem}\label{theo:tree}
For any convex polygon $P$ of $n$ sides, a tree decomposition of $G(P)$ of width at most 3 can be found in $O(n)$ time.
\end{theorem}

\begin{proof}
Let $f$ be the following generic function, that will be used later in the construction of a tree decomposition of $G(P)$: The input consists of three integers $u,v,z\in[0\ldotp\ldotp n]$ and a disk $Q\subset P$, such that: $u<v$, $z\notin[u\ldotp\ldotp v]$, $Q$ is tangent to both $s_u$ and $s_v$, no disk in $\mathcal{D}\setminus(D[u,v]\cup\{D_z\})$ intersects a disk in $D[u+1,v-1]$, and $\ell(s_u)$ and $\ell(s_v)$ are parallel, or for every $k\in [u+1\ldotp\ldotp v-1]$ the line $\ell(s_k)$ separates the interior of $P$ from the intersection point of $\ell(s_u)$ and $\ell(s_v)$. The output $f(u,v,z,Q)$ is a rooted tree decomposition of $G(D[u,v]\cup\{D_z\})$ of width at most 3, whose root node contains $\{D_u,D_v,D_z\}$.

We compute $f(u,v,z,Q)$ as follows: If $u=v-1$, then we return the leaf node $\{D_u,D_v,D_z\}$. Otherwise, if $u<v-1$, then we find an index $t\in[u+1\ldotp\ldotp v-1]$ such that there exists a maximal disk $Q'\subset P$ tangent to $s_u$, $s_t$, and $s_v$. Note that the centers of $Q$ and $Q'$ are the endpoints of an edge of the medial axis of $P$, or coincide. We can then compute both $t$ and $Q'$ in $O(1)$ time from the medial axis. After that we create the root node $r'=\{D_u,D_v,D_z,D_t\}$ of $f(u,v,z,Q)$, which has at most two children, constructed recursively as follows:
\begin{enumerate}
    \item If $u<t-1$, then we set the following child: 
    If the bisector of $\ell(s_u)$ and $\ell(s_t)$ through $P$ intersects $s_z$, 
    then we set the child $f(u,t,z,Q')$. Otherwise, we set the child $f(u,t,v,Q')$.

    \item If $t<v-1$, then we set the following child: 
    If the bisector of $\ell(s_t)$ and $\ell(s_v)$ through $P$ intersects $s_z$, 
    then we set the child $f(t,v,z,Q')$. Otherwise, we set the child $f(t,v,u,Q')$.
\end{enumerate}
Observe that $D[u+1,t-1]$ and $D[t+1,v-1]$ are independent (Lemma~\ref{lem:indep}), which implies that the children of $r'$ are not in conflict with property~(2) of tree decompositions. Furthermore, Lemma~\ref{lem:1-chord} ensures that in each call to $f$ the way we choose the parameters, mainly the third one, is correct.

\begin{figure}[t]
	\centering
	\includegraphics[scale=0.8,page=9]{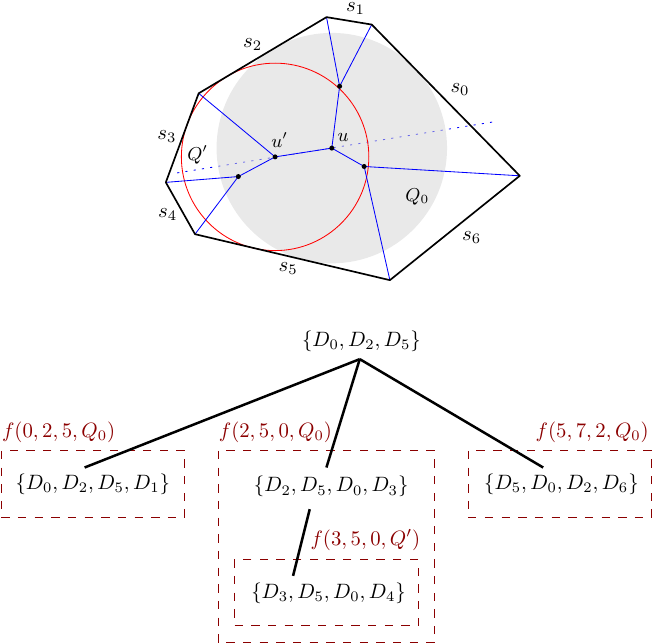}
	\caption{\small{
		Up: A convex 7-gon $P$ of sides $s_0,s_1,\ldots,s_6$. The disks $Q_0$ and $Q'$ are maximal disks
		centered at the vertices $u$ and $u'$ of the medial axis of $P$, respectively. The dotted
		line is the bisector of $\ell(s_2)$ and $\ell(s_5)$, used to call $f(3,5,0,Q')$ from $f(2,5,0,Q_0)$.
		Down: A tree decomposition of $P$ of width 3. Each node different from the root is associated with the
		respective call to the function $f$.
	}}	
	\label{fig:tree-example}
\end{figure}

With $f$ described, we proceed to explain how to construct a rooted tree decomposition $T$ of $G(P)$. Using the medial axis of $P$, we first find a maximal disk $Q_0$ of maximum radius. Note that $Q_0$ can be found in $O(n)$ time. Assume w.l.o.g.\ that $Q_0$ is tangent to the sides $s_0$, $s_i$, and $s_j$, respectively. Further assume $1\le i<j\le n-1$. The root of $T$ is the node $r=\{D_0,D_i,D_j\}$, and we set to $r$ at most 3 children, constructed using $f$ as follows:
\begin{enumerate}
    \item If $i>1$, then we set $f(0,i,j,Q_0)$ as a child of $r$.
    \item If $i<j-1$, then we set $f(i,j,0,Q_0)$ as a child of $r$.
    \item If $j<n-1$, then we set $f(j,n,i,Q_0)$ as a child of $r$.
\end{enumerate}
Refer to Figure~\ref{fig:tree-example} for an example. Since $D[1,i-1]$, $D[i+1,j-1]$, and $D[j+1,n-1]$ are pairwise independent (Lemma~\ref{lem:indep}), this selection of children for $r$ is not in conflict with property~(2) of tree decompositions. Furthermore, Lemma~\ref{lem:1-chord} ensures that in each call to $f$ the parameters, mainly the third one, satisfy the desired properties. The treewidth of $T$ is trivially at most 3, and using the above arguments, it can be proved formally by induction that $T$ is in fact a tree decomposition of $G(P)$. Observe that the running time can be charged to the number of nodes of $T$. Except of the root node of $T$, every other node contains a different side disk of $\mathcal{D}\setminus\{D_0,D_i,D_j\}$. Then, $T$ has $n-2$ nodes and the running time is $O(n)$.\qed
\end{proof}

\section{Bound of the size of a MIS}\label{sec:size-MIS}

In this section, we show that for every $n\ge 3$, there exist convex polygons $P$ of $n$ sides such that $\alpha(G(P))=\lceil n/4\rceil$. The idea is that in those polygons $P$ the vertices of $G(P)$ can be partitioned into $\lceil n/4\rceil$ subsets, such that the induced subgraph of each subset is isomorphic to $K_4$, except possibly one of them that is isomorphic to some $K_r$ with $r\in\{1,2,3\}$. Then, any independent set of $G(P)$ cannot contain two vertices of the same subset; hence it must contain at most $\lceil n/4\rceil$ vertices.

\begin{lemma}\label{lem:a}
Let $P$ be a convex polygon of $n\ge 3$ sides such that $P$ has an acute interior angle and $\mathcal{D}(P)$ is proper. For every $m\in\{1,2,3,4\}$, we can construct from $P$ a convex polygon $Q$ of $n+m$ sides, where the side disks of $Q$ can be partitioned into two sets $\mathcal{D}_0$ and $\mathcal{D}_m$ such that $G(\mathcal{D}_0)\simeq G(P)$ and $G(\mathcal{D}_m)\simeq K_m$. If $m=4$, $Q$ has also an acute interior angle and $\mathcal{D}(Q)$ is proper. 
\end{lemma}

\begin{proof}
Let $\mathcal{V}$ be the vertex set of $P$. Let $A,B,C$ be three consecutive vertices of $P$ in counter-clockwise order along the boundary of $P$, such that the interior angle at $B$ is acute. Since $\mathcal{D}(P)$ is proper, we can select points $A'$ and $C'$ in the interiors of $AB$ and $BC$, respectively, such that: $|A'B|=|BC'|$, and $A'$ and $C'$ are close enough to $B$ so that $G(P)=G(\mathcal{D})\simeq G(\mathcal{D}_0)$, where $\mathcal{D}_0=(\mathcal{D}(P)\setminus\{D_{AB},D_{BC}\})\cup\{D_{AA'},D_{C'C}\}$, and $\mathcal{D}_0$ is also proper. 

Let $Q_1$ be the convex polygon with vertex set $(\mathcal{V}\setminus\{B\})\cup\{A',C'\}$.
If $m=1$, then the polygon $Q=Q_1$ satisfies the desired properties, where $\mathcal{D}_m=\{D_{A'C'}\}$. 

Let $B'\in P$ be a point in the bisector of the angle $\angle ABC$, $B'\neq B$, that is close enough to $B$ so that $\Delta A'B'C'\subset \Delta A'BC'$ and the angle $\angle A'B'C'$ is acute. Let $Q_2$ be the convex polygon with vertex set $(\mathcal{V}\setminus\{B\})\cup\{A',B',C'\}$. If $m=2$, then the polygon $Q=Q_2$ satisfies the desired properties, where $\mathcal{D}_m=\{D_{A'B'},D_{B'C'}\}$.

Let $T$ be a point in the interior of the segment $C'B'$ such that the boundary of $D_{C'T}$ intersects twice the segment $A'B'$ (see Figure~\ref{fig:lem-a}). This latter condition can be ensured because angle $\angle A'B'C'$ is acute and $\Delta A'B'C'$ is isosceles with base $A'C'$. Let $R$ be the intesection point between the boundary of $D_{C'T}$ and $A'B'$ that is closer to $B'$. Let $Q_3$ be the convex polygon with vertex set $(\mathcal{V}\setminus\{B\})\cup\{A',R,T,C'\}$. If $m=3$, then the polygon $Q=Q_3$ satisfies the desired properties, where $\mathcal{D}_m=\{D_{A'R},D_{RT},D_{TC'}\}$.

Let $S$ be a point in the interior of the shortest arc of the boundary of $D_{C'T}$ that connects $R$ and $T$, such that the pair $D_{ST},D_{A'R}$ has a proper intersection (see Figure~\ref{fig:lem-a}). Let $Q_4$ be the convex polygon with vertex set $(\mathcal{V}\setminus\{B\})\cup\{A',R,S,T,C'\}$, and assume $m=4$. Note that the interior angle $\angle STC'$ of $Q_4$ is acute because $\angle C'ST$ equals $\pi/2$ by Thales' theorem. Furthermore, $Q_4$ is proper. Hence, for $m=4$ the polygon $Q=Q_4$ satisfies the desired properties, where $\mathcal{D}_m=\{D_{A'R},D_{RS},D_{ST},D_{TC'}\}$.\qed
\end{proof}

\begin{figure}[t]
	\centering
	\includegraphics[scale=0.7,page=4]{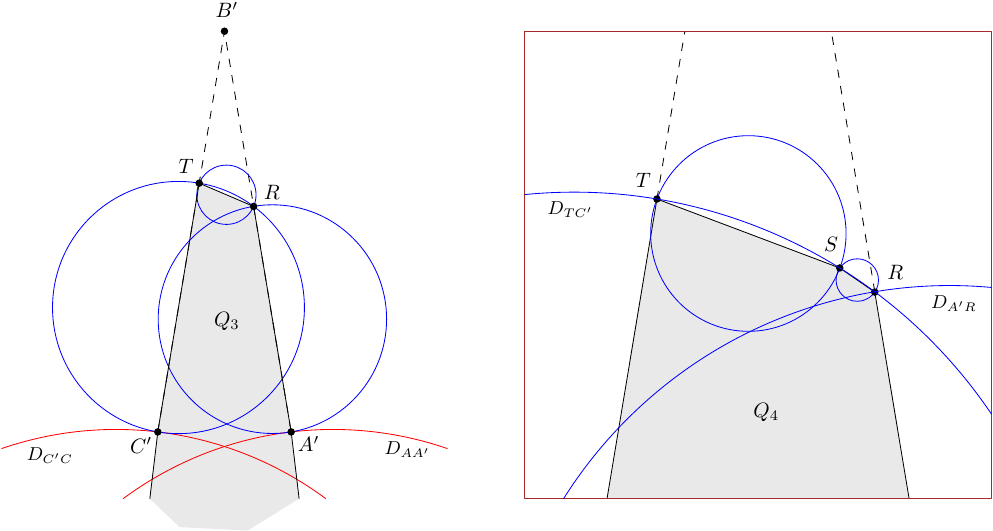}
	\caption{\small{Illustration of Lemma~\ref{lem:a}}.}
	\label{fig:lem-a}
\end{figure}

\begin{theorem}
For every $n\ge 3$, there exists a polygon $P$ of $n$ sides such that $\alpha(G(P))=\lceil n/4\rceil$.
\end{theorem}

\begin{proof}
If $n=3$, then every triangle $P$ satisfies $G(P)\simeq K_3$ and $\alpha(G(P))=1=\lceil 3/4\rceil$. Assume $n\ge 4$, and let $t\ge 1$ and $r\in\{0,1,2,3\}$ be the integer numbers such that $n=4t+r$. Let $P_4$ be a convex quadrilateral with an acute interior angle and such that $\mathcal{D}(P_4)$ is proper and $G(P_4)\simeq K_4$. For example, $P_4$ can be the quadrilateral with vertices at coordinates $(1,0)$, $(0,2)$, $(-1,0)$, and $(0,-1)$, respectively. Consider the sequence $P_4,P_8,P_{12},\ldots,P_{4t}$ of $t$ convex polygons, where for $i=2,\ldots,t$ the polygon $P_{4i}$ has $4i$ sides and is constructed from $P_{4i-4}$ by using Lemma~\ref{lem:a} with $m=4$. By construction, the vertices of $G(P_{4t})$ can be partitioned into $t$ subsets $V_1,V_2,\ldots,V_t$ of size 4, such that the graph induced by each $V_i$ is isomorphic to $K_4$. Since any independent set of $G(P_{4t})$ cannot contain two vertices of the same $V_i$, $\alpha(G(P_{4t}))\le t=\lceil n/4\rceil$.

If $r\in\{1,2,3\}$, let $P_{4t+r}$ be the convex polygon of $n=4t+r$ sides constructed from $P_{4t}$ by using Lemma~\ref{lem:a} with $m=r$. The vertices of $G(P_{4t+r})$ can be partitioned into subsets $W_1,W_2,\ldots,W_t,W_{t+1}$ where the graph induced by each of $W_1,W_2,\ldots,W_t$ is isomorphic to $K_4$, and the graph induced by $W_{t+1}$ is isomorphic to $K_r$. Since any independent set of $G(P_{4t+r})$ cannot contain two vertices of the same $W_i$, we obtain $\alpha(G(P_{4t+r}))\le t+1=\lceil n/4\rceil$. 

For any $r\in\{0,1,2,3\}$, let $P=P_{4t+r}$, which satisfies $\alpha(G(P))\le \lceil n/4\rceil$. Since $G(P)$ is planar, its chromatic number is at most 4. Then, it can be colored with at most 4 colors, and the majority color induces an independent set of size at least $\lceil n/4\rceil$. Hence, $\alpha(G(P))=\lceil n/4\rceil$.\qed
\end{proof}

\section{Classes of planar Hamiltonian graphs}\label{sec:classes}

In this section, we prove that our class of graphs, that is, the side-disk realizable graphs, contains all outerplanar Hamiltonian graphs except the cycle of lengh four; and that it is a proper subclass of the planar Hamiltonian graphs. Let $C_4$ denote the cycle graph of length four.

\begin{lemma}\label{lem:not-C4}
The graph $C_4$ is not side-disk realizable.
\end{lemma}

\begin{proof}
It suffices to prove that in any (convex) quadrilateral of sides denoted $s_0$, $s_1$, $s_2$, and $s_3$ counter-clockwise, $D_0\cap D_2\neq\emptyset$ or $D_1\cap D_3\neq\emptyset$. Let $x$ be the distance between the midpoints of $s_1$ and $s_3$, and $y$ be the distance between the midpoints of $s_0$ and $s_2$. Observe that $D_0\cap D_2\neq\emptyset$ if and only if $|s_0|+|s_2|\ge 2y$. Similarly, $D_1\cap D_3\neq\emptyset$ if and only if $|s_1|+|s_3|\ge 2x$. Since in any quadrilateral the total lenght of two opposite sides is at least twice the distance between the midpoints of the other two sides~\cite[Innequality 15.10]{bottema1969geometric}, we have $|s_0|+|s_2|\ge 2x$ and $|s_1|+|s_3|\ge 2y$. That is, $|s_0|+|s_2|+|s_1|+|s_3|\ge 2x+2y$. Hence, we must have $|s_0|+|s_2|\ge 2y$ or $|s_1|+|s_3|\ge 2x$, which implies the result.\qed
\end{proof}

\begin{lemma}\label{lem:b}
Let points $A$, $B$, and $C$ be the vertices of an isosceles triangle of base $AC$, so that the interior angle at $B$ is at least $\pi/3$. Let $A'\in\ell(A,B)$ and $C'\in\ell(B,C)$ be points such that $A$ belongs to the interior of $A'B$, and $C$ belongs to the interior of $BC'$. Let $k\ge 2$ be an integer, and for every $i\in[0\ldotp\ldotp k]$, let $Q_i$ denote the point of $AC$ such that $|CQ_i|=(i/k)\cdot|AC|$. Then, $D_{Q_iQ_{i+1}}\cap D_{AA'}=\emptyset$ for every $i\in[0\ldotp\ldotp k-2]$, and $D_{Q_iQ_{i+1}}\cap D_{CC'}=\emptyset$ for every $i\in[1\ldotp\ldotp k-1]$.
\end{lemma}

\begin{figure}[t]
	\centering
	\includegraphics[scale=0.9,page=11]{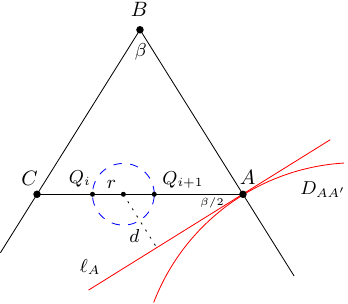}
	\caption{\small{Proof of Lemma~\ref{lem:b}.}}
	\label{fig:lem-b}
\end{figure}

\begin{proof}
We will prove the first part of the lemma. The proof for the second one is analogous. Let $\ell_A$ be the line through $A$ and perpendicular to $\ell(A,B)$. Note that $\ell_A$ separates the points $Q_i$, $i\in[0\ldotp\ldotp k-1]$, from the disk $D_{AA'}$ (see Figure~\ref{fig:lem-b}). Let $b=|AC|$, $\beta=\angle ABC\ge \pi/3$, and $i\in[0\ldotp\ldotp k-2]$. Note that the radius $r$ of $D_{Q_iQ_{i+1}}$ equals $b/(2k)$, and the distance $d$ from the midpoint of $Q_iQ_{i+1}$ to $\ell_A$ equals
\begin{align*} 
(b - |CQ_i|-r) \cdot \sin(\beta/2) & ~=~  \left( b - \frac{ib}{k} - r \right) \cdot \sin(\beta/2) \\
    & ~\ge~ \frac{1}{2}\left( b - \frac{(k-2)b}{k} - r \right) \\
    & ~=~  \frac{1}{2}\left(\frac{2b}{k}-r \right)
            ~=~ \frac{3}{2}r ~>~ r.
\end{align*}
Since $d>r$, the line $\ell_A$ separates $D_{Q_iQ_{i+1}}$ from the interior of $D_{AA'}$. This implies $D_{Q_iQ_{i+1}}\cap D_{AA'}=\emptyset$ for every $i\in[0\ldotp\ldotp k-2]$.\qed
\end{proof}

\begin{theorem}\label{theo:outerplanar}
Every outerplanar Hamiltonian graph, except $C_4$, is side-disk realizable. 
\end{theorem}

\begin{proof}
We will prove the theorem by induction in the number of faces of the planar representation of the graph. More precisely, we prove that: For every outerplanar Hamiltonian graph $G$ different from $C_4$, there exists a convex polygon $P$ such that $G\simeq G(P)$, and $P$ is good. We say that a convex polygon $P$ is {\em good} if $\mathcal{D}(P)$ is proper, every vertex of $P$ is contained only in the side disks of the sides of $P$ adjacent to the vertex, and all interior angles of $P$ are at least $\pi/3$.

The idea behind the induction is to start with realizing an internal face, or two neighboring internal faces, of the planar representation of the graph as a side-disk intersection graph. Then, to incorporate to the realization a new internal face of the planar representation that is a neighbor of a face already realized; until all internal faces have been incoporated. If there are internal faces of the planar embedding with 3 or at least 5 vertices, then we start with any of them. Otherwise, if all internal faces have exactly 4 vertices, then there are at least two of them because the graph is not $C_4$, and we start with realizing any two neighboring internal faces.

Let $G=(V,E)$ be an outerplanar Hamiltonian graph different from $C_4$. We have two base cases in our induction: (1) $G$ is a cycle. Then $|V|=3$ or $|V|\ge 5$, and any regular polygon $P$ of $|V|$ vertices satisfies the desired properties. Note that this is not true when $|V|=4$ by Lemma~\ref{lem:not-C4}. (2) $G$ has $3$ faces, one external of $|V|=6$ vertices, and two internals of four vertices each that share and edge of $G$. Then, the hexagon $P$ with vertex set $\{(-1,0),(0,-1),(1,0),(1,3),(0,4),(-1,3)\}$ satisfies the desired properties.

Consider now the general case, and let $f>2$ be the number of faces of the planar representation of $G$. Note that $G$ can be decomposed into two induced subgraphs $G_0$ and $G_1$, where $G_0$ is an outerplanar Hamiltonian graph of $f-1$ faces, and $G_1$ is a path $(v_0,v_1,\ldots,v_t,v_{t+1})$, $t\ge 1$, such that $v_0$ and $v_{t+1}$ are consecutive vertices of the outer face of $G_0$, and $v_1,\ldots,v_t$ are not vertices of $G_0$. Using the inductive hypothesis, let $P_0$ be a convex polygon such that $G(P_0)\simeq G_0$ and $P_0$ is good. Let $\mathcal{V}_0$ denote the vertex set of $P_0$. Let $g$ be an isomorphism between $G(P_0)$ and $G_0$, and let $A',B,C'$ be three consecutive vertices of $P_0$ such that $g(v_0)=D_{BC'}$ and $g(v_{t+1})=D_{A'B}$.

Since $\mathcal{D}(P_0)$ is good, we can select points $A$ and $C$ in the interiors of $A'B$ and $BC'$, respectively, such that: $|AB|=|BC|$, and $A$ and $C$ are close enough to $B$ so that the polygon $P'_0$ with vertex set $(\mathcal{V}_0\setminus \{B\})\cup\{A,C\}$ is good; $G(\mathcal{D}(P'_0)\setminus\{D_{AC}\})\simeq G_0$ for some isomorphism $h$ such that $h(D_{AA'})=D_{A'B}$ and $h(D_{CC'})=D_{BC'}$; and the side disk $D_{AC}$ of $P'_0$ only intersects its two neighbor side disks $D_{AA'}$ and $D_{CC'}$ (refer to Figure~\ref{fig:outerplanar}). If $t=1$, note that $P=P'_0$ satisfies $G(P)\simeq G$. Otherwise, if $t\ge 2$, we proceed as follows: Let $Q_1,Q_2,\ldots,Q_{t-1}\in AC$ be $t-1$ points that split the segment $AC$ into $t$ segments of equal length, and they appear in this order from $C$ to $A$. Let $Q_0=C$ and $Q_t=A$. By Lemma~\ref{lem:b}, we have that the disk sets $\{D_{Q_iQ_{i+1}}\mid i\in[1\ldotp\ldotp t-1]\}$ and $\mathcal{D}(P'_0)\setminus\{D_{AA'},D_{AC},D_{CC'}\}$ are independent. Furthermore, the intersection graph of $\{D_{C'C},D_{Q_0Q_1},\ldots,D_{Q_{t-1}Q_t},D_{AA'}\}$ is a path of length $t+1$, connecting $D_{CC'}=h^{-1}(D_{BC'})$ with $D_{AA'}=h^{-1}(D_{A'B})$, then isomorphic to $G_1$. This implies that for $\mathcal{D}=(\mathcal{D}(P'_0)\setminus\{D_{AC}\})\cup \{D_{Q_0Q_1},\ldots,D_{Q_{t-1}Q_t}\}$ we have $G(\mathcal{D})\simeq G$. Now, observe that there exist points $R_1,R_2,\ldots,R_{t-1}$ close enough to $Q_1,Q_2,\ldots,Q_{t-1}$, respectively, so that the point set $(\mathcal{V}_0\setminus \{B\})\cup\{A,C,Q_1,\ldots,Q_{t-1}\}$ is the vertex set of a good convex polygon $P$ that ensures $G(P)\simeq G(\mathcal{D})\simeq G$. The theorem thus follows.\qed
\end{proof}

\begin{figure}[t]
	\centering
	\includegraphics[scale=0.75,page=10]{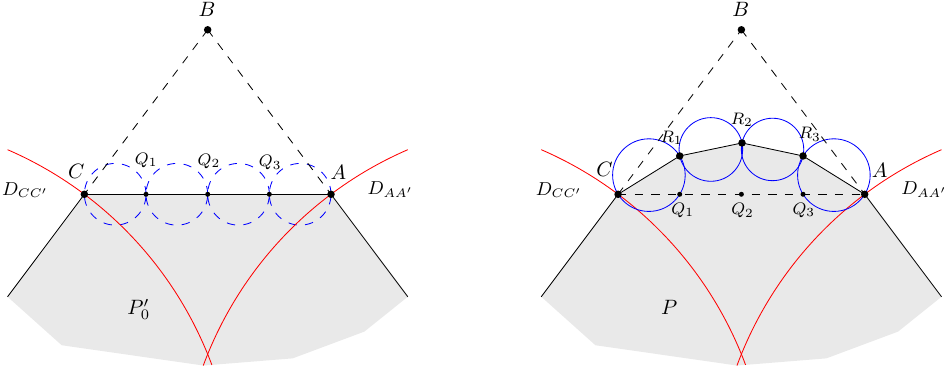}
	\caption{\small{Proof of Theorem~\ref{theo:outerplanar}. Left: The points $Q_1,Q_2,Q_3$ for $t=4$. 
	Right: The points $R_1,R_2,R_3$ and the polygon $P$ such that $G(P)\simeq G$.}}
	\label{fig:outerplanar}
\end{figure}

\begin{theorem}
Let $\mathcal{G}$ be the class of the side-disk realizable graphs. Then, $\mathcal{G}$ contains all outerplanar Hamiltonian graphs except $C_4$, there are graphs in $\mathcal{G}$ which are not outerplanar, and $\mathcal{G}$ is a proper subclass of the planar Hamiltonian graphs.
\end{theorem}

\begin{proof}
Theorem~\ref{theo:outerplanar} implies that $\mathcal{G}$ contains all outerplanar Hamiltonian graphs except the cycle graph $C_4$. To see that there are graphs in $\mathcal{G}$ which are not outerplanar, consider any graph of $\mathcal{G}$  with $K_4$ as a minor (see for example Section~\ref{sec:size-MIS}, or Figure~\ref{fig:example}). Such a graph has treewidth 3, whereas all outerplanar graphs have treewidth 2. Finally, to see that $\mathcal{G}$ is a proper subclass of the planar Hamiltonian graphs, consider again the $k\times k$ grid graph of $n=k^2$ vertices, with $k\ge 4$ and even, which is both planar and Hamiltonian. Since its treewidth is $k=\sqrt{n}>3$, it does not belong to $\mathcal{G}$.\qed
\end{proof}

From the beggining of the study of the side-disk realizable graphs, also in this paper, the side disks have been considered as closed disks~\cite{huemer2016}. If we consider the side disks as open disks, the intersection graph of them is also planar and Hamiltonian. Furthermore, the cycle graph $C_4$ is side-disk realizable by a square, and all outerplanar Hamiltonian graphs are then side-disk realizable. To note this last statement, observe that in the realization constructed in the proof of Theorem~\ref{theo:outerplanar} all pairs of intersecting disks have a proper intersection, that is, the open versions of the disks have a non-empty intersection. Hence, in this case the class of the side-disk realizable graphs is a proper superclass of the outerplanar Hamiltonian graphs, and a proper subclass of the planar Hamiltonian graphs.

\small

\section*{Acknowledgements}

We wish to thank the {\tt GeoGebra} open source software and its developers~\cite{gg}. This research has been supported by projects CONICYT FONDECYT/Regular 1160543 (Chile), and Millennium Nucleus Information and Coordination in Networks ICM/FIC RC130003 (Chile).

\small

\bibliographystyle{abbrv}
\bibliography{refs}

\begin{thebibliography}{10}

\bibitem{bodlaender1988dynamic}
H.~L. Bodlaender.
\newblock Dynamic programming on graphs with bounded treewidth.
\newblock {\em Automata, Languages and Programming}, pages 105--118, 1988.

\bibitem{bodlaender1994tourist}
H.~L. Bodlaender.
\newblock A tourist guide through treewidth.
\newblock {\em Acta Cybernetica}, 11(1-2):1--21, 1993.

\bibitem{bodlaender1996linear}
H.~L. Bodlaender.
\newblock A linear-time algorithm for finding tree-decompositions of small
  treewidth.
\newblock {\em SIAM Journal on Computing}, 25(6):1305--1317, 1996.

\bibitem{bottema1969geometric}
O.~Bottema, R.~Z. Djordjevi\'c, R.~R. Jani\'c, D.~S. Mitrinovi\'c, and P.~M.
  Vasi\'c.
\newblock {\em Geometric inequalities}.
\newblock Wolters-Noordhoff, 1969.

\bibitem{chin1999finding}
F.~Chin, J.~Snoeyink, and C.~A. Wang.
\newblock Finding the medial axis of a simple polygon in linear time.
\newblock {\em Discrete \& Computational Geometry}, 21(3):405--420, 1999.

\bibitem{clark1990}
B.~N. Clark, C.~J. Colbourn, and D.~S. Johnson.
\newblock Unit disk graphs.
\newblock {\em Discrete Mathematics}, 86(1-3):165--177, 1990.

\bibitem{fleischner2010}
H.~Fleischner, G.~Sabidussi, and V.~I. Sarvanov.
\newblock Maximum independent sets in 3-and 4-regular hamiltonian graphs.
\newblock {\em Discrete Mathematics}, 310(20):2742--2749, 2010.

\bibitem{gg}
M.~Hohenwarter.
\newblock {G}eo{G}ebra: Ein {S}oftwaresystem f\"ur dynamische {G}eometrie und
  {A}lgebra der {E}bene.
\newblock Master's thesis, Paris Lodron University, Salzburg, Austria, 2002.
\newblock (In German.).

\bibitem{huemer2016}
C.~Huemer and P.~P{\'e}rez-Lantero.
\newblock The intersection graph of the disks with diameters the sides of a
  convex $ n $-gon.
\newblock {\em arXiv preprint arXiv:1410.4126}, 2016.

\bibitem{Lokshtanov2011}
D.~Lokshtanov, D.~Marx, and S.~Saurabh.
\newblock Known algorithms on graphs of bounded treewidth are probably optimal.
\newblock In {\em Proceedings of the Twenty-second Annual ACM-SIAM Symposium on
  Discrete Algorithms}, SODA '11, pages 777--789, 2011.

\bibitem{preparata1977}
F.~P. Preparata.
\newblock The medial axis of a simple polygon.
\newblock In {\em Mathematical Foundations of Computer Science 1977}, volume~53
  of {\em LNCS}, pages 443--450. 1977.

\bibitem{swart1983}
G.~Swart and R.~Ladner.
\newblock {\em Efficient algorithms for reporting intersections}.
\newblock Department of Computer Science, University of Washington, 1983.

\bibitem{wang1988}
D.~Wang and Y.-S. Kuo.
\newblock A study on two geometric location problems.
\newblock {\em Information Processing Letters}, 28(6):281--286, 1988.

\end{thebibliography}

\end{document}